\documentclass[letterpaper,11pt]{article}

\newcommand{\longversion}[1]{#1}
\newcommand{\shortversion}[1]{}
\newcommand{\longshort}[2]{\longversion{#1}\shortversion{#2}}

\usepackage{amsthm}

\usepackage{amsmath,amssymb,amstext}

\usepackage{url} 
\usepackage{booktabs}
\usepackage{threeparttable}
\usepackage{tabularx}

\usepackage{tikz}
\usepackage{boxedminipage}
\usepackage{xspace}
\usepackage{wasysym}

\newcommand{\SB}{\{\,}%
\newcommand{\SM}{\;{:}\;}%
\newcommand{\SE}{\,\}}%

\usepackage{subfig}

\usepackage[letterpaper,margin=1in]{geometry}

\newtheorem{definition}{Definition}
\newtheorem{lemma}{Lemma} 

\newtheorem{theorem}{Theorem}

\newtheorem{myrule}{Rule}
\newtheorem{claim}{Claim}

\newcommand{\hy}{\hbox{-}\nobreak\hskip0pt}

\newcommand{\Card}[1]{|#1|}

 \let\AAA\cA
\newcommand{\cC}{\mathcal{C}} \let\CCC\cC

 \let\GGG\cG

\newcommand{\cO}{\mathcal{O}}
\newcommand{\cB}{\mathcal{B}}
 
\newcommand{\TW}{\mathcal{W}}

\newcommand{\NP}{\text{\normalfont NP}}

\newcommand{\set}[1]{\left\{ #1 \right\}}
\newcommand{\myiff}{iff\xspace}

\newcommand{\ie}{i.e.\xspace}
\newcommand{\etal}{\emph{et al.}\xspace}

\newcommand{\var}{{\normalfont \textsf{var}}}
\newcommand{\cla}{{\normalfont \textsf{cla}}}
\newcommand{\lit}{{\normalfont \textsf{lit}}}
\newcommand{\true}{1} 
\newcommand{\false}{0} 

\newcommand{\inc}{{\normalfont \textsf{inc}}}

\newcommand{\tw}{{\normalfont \textsf{tw}}}
\renewcommand{\sb}{{\normalfont \textsf{sb}}}

\newcommand{\class}[1]{\text{\text{\normalfont\sc  #1}}}

\newcommand{\RHorn}{\class{RHorn}}

\newcommand{\Treewidth}{\ensuremath{\TW_{\leq t}}\xspace}

\newcommand{\BDS}{backdoor set\xspace}
\newcommand{\BDSs}{backdoor sets\xspace}

\newcommand{\TBDS}{\Treewidth{}\hy \BDS}
\newcommand{\TBDSs}{\Treewidth{}\hy \BDSs}
\newcommand{\STBDS}{strong \TBDS}
\newcommand{\STBDSs}{strong \TBDSs}

\newcommand{\Wobs}{wall-ob\-struc\-tion\xspace}
\newcommand{\Wobss}{wall-ob\-struc\-tions\xspace}

\newcommand{\obstemp}{obstruction-template\xspace}
\newcommand{\obstemps}{obstruction-templates\xspace}
\newcommand{\OT}{\mathsf{OT}}

\newcommand{\ktw}{\mathsf{tw}(k,t)}
\newcommand{\kwall}{\mathsf{wall}(k,t)}
\newcommand{\kobs}{\mathsf{obs}(k,t)}
\newcommand{\ksame}{\mathsf{same}(k,t)}
\newcommand{\tnb}{\mathsf{nb}(t)}

\allowdisplaybreaks

\title{Strong Backdoors to Bounded Treewidth SAT %
\thanks{Research supported by the European Research Council (ERC), project COMPLEX REASON 239962.}}
\author{%
Serge Gaspers \and
Stefan Szeider}

\date{%
\normalsize Institute of Information Systems\\ Vienna University of Technology\\ Vienna, Austria.\\
\texttt{gaspers@kr.tuwien.ac.at}\\ \texttt{stefan@szeider.net}
}

\begin{document}

\maketitle

 \thispagestyle{empty}

\begin{abstract}
  There are various approaches to exploiting ``hidden structure'' in
  instances of hard combinatorial problems to allow faster algorithms
  than for general unstructured or random instances. For SAT and its
  counting version \#SAT, hidden structure has been exploited in terms
  of decomposability  and strong backdoor sets.
    Decomposability can be considered in terms of the treewidth of
  a graph that is associated with the given CNF formula, for instance
  by considering clauses and variables as vertices of the graph, and
  making a variable adjacent with all the clauses it appears in.
  On the other hand, a strong backdoor set of a CNF formula is a set of
  variables such that each possible partial assignment to this set
  moves the formula into a fixed class  for which (\#)SAT
  can be solved in polynomial time.

  In this paper we combine the two above approaches. In particular, we
  study the algorithmic question of finding a small strong backdoor
  set into the class $\TW_{\leq t}$ of CNF formulas whose associated
  graphs have treewidth at most~$t$. The main results are positive:

  \begin{enumerate}
  \item[(1)] There is a cubic-time
    algorithm that, given a CNF formula $F$ and two constants $k,t\ge 0$, either finds a strong
    $\TW_{\leq t}$\hy backdoor set of size at most $2^k$, or concludes
    that $F$ has no strong $\TW_{\leq t}$\hy backdoor set of
    size at most~$k$.

  \item[(2)] There is a cubic-time algorithm that,
      given a CNF formula $F$,
      computes the number of satisfying assignments of $F$
      or concludes that $\sb_t(F)>k$, for any pair of constants $k,t\ge 0$.
      Here, $\sb_t(F)$ denotes the size of a smallest strong $\TW_{\leq t}$\hy backdoor
      set of~$F$.
  \end{enumerate}
  We establish both results by distinguishing between two cases,
  depending on whether the treewidth of the given formula is small or
  large. For both results the case of small treewidth
  can be dealt with relatively standard  methods.
  The case of large treewidth is
  challenging and requires new and sophisticated combinatorial
  arguments. The main tool is an auxiliary graph
  whose vertices represent subgraphs in $F$'s associated graph.
  It captures various ways to assemble
  large-treewidth subgraphs in $F$'s associated graph.
  This is used to show that every backdoor set of size $k$
  intersects a certain set of variables whose size is bounded
  by a function of $k$ and $t$. For any other set of $k$ variables,
  one can use the auxiliary graph to find an assignment $\tau$ to these variables
  such that the graph associated with $F[\tau]$ 
  has treewidth at least $t+1$.

  The significance of our results lies in the fact that they allow us
  to exploit algorithmically a hidden structure in formulas that is
  not accessible by any one of the two approaches (decomposability,
  backdoors) alone.  Already a backdoor size 1 on top of treewidth 1
  (i.e., $\sb_1(F)=1$) entails formulas of arbitrarily large treewidth
  and arbitrarily large cycle cutsets (variables that need to be deleted to
  make the instance acyclic).
  
\end{abstract}

\newpage
\setcounter{page}{1}
\section{Introduction}

\paragraph{Background.}
Satisfiability (SAT) is probably one of the most important NP-complete
problems~\cite{Cook71,Levin73}. Despite the theoretical intractability of SAT,
heuristic algorithms work surprisingly fast on real-world SAT instances. A
common explanation for this discrepancy between theoretical hardness and
practical feasibility is the presence of a certain ``hidden structure''
in industrial SAT instances~\cite{GomesKautzSabharwalSelman08}. There are
various approaches to capturing the vague notion of a ``hidden
structure'' with a mathematical concept.

One widely studied approach is to consider the hidden structure in terms
of \emph{decomposability}. The basic idea is to decompose a SAT instance
into small parts that can be solved individually, and to put solutions
for the parts together to a global solution. The overall complexity
depends only on the maximum overlap of the parts, the \emph{width} of
the decomposition. Treewidth and branchwidth are two decomposition width
measures (related by a constant factor) that have been applied to
satisfiability. The width measures are either applied in terms of the
\emph{primal graph} of the formula (variables are vertices, two
variables are adjacent if they appear together in a clause) or in terms
of the \emph{incidence graph} (a bipartite graph on the variables and
clauses, a clause is incident to all the variables it contains).  If the
treewidth or branchwidth of any of the two graphs is bounded, then SAT
can be decided in polynomial time; in fact, one can even count the
number of satisfying assignments in polynomial time. This result has
been obtained in various contexts, e.g., resolution
complexity~\cite{AlekhnovichRazborov02} and Bayesian Inference~\cite{BacchusDalmaoPitassi03} (branchwidth of primal graphs), and Model
Checking for Monadic Second-Order Logic~\cite{FischerMakowskyRavve06}
(treewidth of incidence graphs).

A complementary approach is to consider the hidden structure of a SAT
instance in terms of a small set of key variables, called \emph{backdoor
  set}, that when instantiated moves the instance into a polynomial
class. More precisely, a \emph{strong} backdoor set of a CNF formula $F$
into a polynomially solvable class $\CCC$ (or strong $\CCC$\hy backdoor
set, for short) is a set $B$ of variables such that for all partial
assignments $\tau$ to $B$, the reduced formula $F[\tau]$ belongs
to~$\CCC$ (\emph{weak} backdoor sets apply only to satisfiable formulas and
will not be considered in this paper).  Backdoor sets where introduced
by Williams \etal~\cite{WilliamsGomesSelman03} to explain favorable
running times and the heavy-tailed behavior of SAT and CSP solvers on
practical instances. In fact, real-world instances tend to have small
backdoor sets (see \cite{LiB11} and references). Of special interest are base classes for which we can
find a small backdoor set efficiently, if one exists. This is the case,
for instance, for the base classes based on the tractable cases in
Schaefer's dichotomy theorem~\cite{Schaefer78}. In fact, for any
constant $b$ one can decide in linear time whether a given CNF formula
admits a backdoor set of size $b$ into any Schaefer
class~\cite{GaspersSzeider11b}.

\paragraph{Contribution.}
In this paper we combine the two above approaches. In particular, we
study the algorithmic question of finding a small strong backdoor set
into a class of formulas of bounded treewidth. Let $\TW_{\leq t}$ denote
the class of CNF formulas whose incidence graph has treewidth at
most~$t$.  Since SAT and \#SAT can be solved in linear time for formulas
in $\TW_{\leq t}$ \cite{FischerMakowskyRavve06,SamerSzeider10}, we can 
also solve these problems efficiently for a formula $F$ if we know a strong
$\TW_{\leq t}$\hy backdoor set of $F$ of small size $k$. We simply take the sum
of the satisfying assignments over all $2^k$ reduced formulas that we
obtain by applying partial truth assignments to a backdoor set of
size~$k$.

However, finding a small strong backdoor set into a class $\TW_{\leq t}$ is
a challenging problem. What makes the problem difficult is that applying
partial assignments to variables is a much more powerful operation than
just deleting the variables from the formula, as setting a variable to
true may remove a large set of clauses, setting it to false removes a
different set of clauses, and for a strong backdoor set $B$ we must
ensure that for all the $2^{\Card{B}}$ possible assignments the
resulting formula is in $\TW_{\leq t}$. The brute force algorithm
tries out all possible sets $B$ of at most $k$ variables, and checks for
each set whether all the $2^{|B|}$ reduced formulas belong to $\TW_{\leq
  t}$. The number of membership checks is of order $2^kn^k$ for an input
formula with $n$ variables. This number is polynomial for constant $k$,
but the order of the polynomial depends on the backdoor size~$k$.  Is it
possible to get $k$ out of the exponent and to have the same polynomial
for every fixed $k$ and $t$? Our main result provides an affirmative
answer to this question. We show the following.
\begin{theorem}\label{the:main}
  There is a cubic-time algorithm
  that, given a CNF formula $F$ and two constants $k,t\ge 0$, either finds a strong $\TW_{\leq t}$\hy
  backdoor set of size at most $2^k$, or concludes that $F$ has no
  strong $\TW_{\leq t}$\hy backdoor set of size at most~$k$.
\end{theorem}

\noindent
Our algorithm distinguishes for a given CNF formula between two cases:
(A)~the formula has small treewidth, or (B)~the formula has large
treewidth. In Case~A we use model checking for monadic second order
logic~\cite{ArnborgLagergrenSeese91} to find a smallest backdoor set.
Roberson and Seymour's theory of graph minors~\cite{RobertsonSeymour85}
guarantees a finite set of forbidden minors for every minor-closed class of
graphs. Although their proof is non-constructive, for the special case of
bounded treewidth graphs the forbidden minors can be computed in constant time
\cite{AdlerGK08,Lagergren98}. These forbidden minors are used in our monadic
second order sentence to describe a strong backdoor set to the base class $\TW_{\leq t}$.
A model checking algorithm \cite{ArnborgLagergrenSeese91} then computes
a strong $\TW_{\leq t}$\hy backdoor set of size~$k$ if one exists.

In Case~B we use a theorem by
Robertson and Seymour~\cite{RobertsonSeymour86b}, guaranteeing a large
wall as a topological minor, to find many
vertex-disjoint obstructions in the incidence graph, so-called \Wobss.
A backdoor set needs to
``kill'' all these obstructions, where an obstruction is killed either internally
because it contains a backdoor variable, or externally because it contains two clauses
containing the same backdoor variable with opposite signs.
Our main combinatorial tool is the \obstemp, a bipartite graph
with external killers on one side and vertices representing vertex-disjoint connected subgraphs
of a \Wobs on the other side of the bipartition.
It is used to guarantee that for sets of $k$ variables excluding a bounded set of variables,
every assignment to these $k$ variables produces a formula whose incidence graph
has treewidth at least $t+1$.

Combining both cases leads to an algorithm producing a
strong $\TW_{\leq t}$\hy backdoor set of a given formula $F$ of size at
most $2^k$ if $F$ has a strong $\TW_{\leq t}$\hy backdoor set of
size~$k$. 

For our main applications of
Theorem~\ref{the:main}, the problems SAT and \#SAT, we can
solve Case~A actually without recurring to the list of forbidden minors of bounded treewidth graphs
and to model checking for monadic second order
logic. Namely, when
the treewidth of the incidence graph is small, we can directly
apply one of the known linear-time algorithms to count the number of
satisfying truth
assignments~\cite{FischerMakowskyRavve06,SamerSzeider10}, thus 
avoiding the issue of finding a backdoor set.

We arrive at the following statement where 
$\sb_t(F)$ denotes the size of a smallest strong $\TW_{\leq t}$\hy backdoor
set of a formula~$F$.
\begin{theorem}\label{the:counting}
 There is a cubic-time algorithm that, given a CNF formula $F$,
 computes the number of satisfying assignments of $F$
 or concludes that $\sb_t(F)>k$, for any pair of constants $k,t\ge 0$.
\end{theorem}

\noindent
This is a \emph{robust} algorithm in the sense of \cite{Spinrad03}
since for every instance, it either solves the problem (SAT, \#SAT)
or concludes that the instance is not in the class of instances that
needs to be solved (the CNF formulas $F$ with $\sb_t(F)\le k$).
In general, a robust algorithm solves the problem on a superclass
of those instances that need to be solved, and it does not necessarily
check whether the given instance is in this class.


Theorem~\ref{the:counting} applies to formulas of \emph{arbitrarily large treewidth}.
We would like to illustrate this with the following example.
Take a CNF formula $F_n$ whose incidence graph is obtained from an
$n\times n$ square grid containing all the variables of $F_n$ by subdividing
each edge by a clause of $F_n$.
It is well-known that the $n\times n$ grid, $n\ge 2$, has treewidth
$n$ and that a subdivision of an edge does not decrease the treewidth.
Hence $F_n\notin \TW_{\leq n-1}$. Now take a new
variable $x$ and add it positively to all horizontal clauses and negatively to
all vertical clauses. Here, a \emph{horizontal} (respectively, a \emph{vertical}) clause
is one that subdivides a horizontal (respectively, a vertical) edge in the natural
layout of the grid. Let $F_n^x$ denote the new formula. Since the
incidence graph of $F_n$ is a subgraph of the incidence graph of
$F_n^x$, we have $F_n^x\notin \TW_{\leq n-1}$. However, setting~$x$ to
true removes all horizontal clauses and thus yields a formula whose incidence
graph is acyclic, hence $F_n^x[x=\text{true}]\in \TW_{\leq 1}$.
Similarly, setting~$x$ to false yields a formula
$F_n^x[x=\text{false}]\in \TW_{\leq 1}$. Hence $\{x\}$ forms a strong
$\TW_{\leq 1}$\hy backdoor set of $F_n^x$. Conversely, it is easy to
construct, for every $t\ge 0$, formulas that belong to $\TW_{\leq t+1}$ but
require arbitrarily large strong $\TW_{\leq t}$\hy backdoor sets.

One can also define a \emph{deletion $\CCC$\hy backdoor set} $B$ of a
CNF formula $F$ by requiring that deleting all literals $x,\neg x$ with
$x\in B$ from~$F$ produces a formula that belongs to the base
class~\cite{NishimuraRagdeSzeider07}. For many base classes it holds
that every deletion backdoor set is a strong backdoor set, but in most
cases, including the base class $\TW_{\leq t}$, the reverse is not
true. In fact, it is easy to see that if a CNF formula $F$ has a
deletion $\TW_{\leq t}$\hy backdoor set of size $k$, then $F\in
\TW_{t+k}$. In other words, the parameter ``size of a smallest deletion
$\TW_{\leq t}$\hy backdoor set'' is dominated by the parameter
``treewidth of the incidence graph'' and therefore of limited theoretical
interest, except for reducing the space requirements of dynamic
programming procedures~\cite{BidyukDechter07} and analyzing the
effectiveness of polynomial time preprocessing~\cite{CyganLPPS11}.

A common approach to solve \#SAT is to find a small \emph{cycle cutset}
(or feedback vertex set) of variables of the given CNF formula, and by
summing up the number of satisfying assignments of all the acyclic
instances one gets by setting the cutset variables in all possible
ways~\cite{Dechter03}.  We would like to note that such a cycle cutset
is nothing but a deletion $\TW_{\leq 1}$\hy backdoor set. By considering
strong $\TW_{\leq 1}$\hy backdoor sets instead, one can get
super-exponentially smaller sets of variables, and hence a more powerful
method.  A strong $\TW_{\leq 1}$\hy backdoor set can be considered as a
an \emph{implied cycle cutset} as it cuts cycles by removing clauses
that are satisfied by certain truth assignments to the backdoor
variables. By increasing the treewidth bound from $1$ to some fixed
$t>1$ one can further dramatically decrease the size of a smallest
backdoor set.

Our results can also be phrased in terms of \emph{Parameterized
  Complexity} \cite{FlumGrohe06}.  Theorem~\ref{the:counting} states
that \#SAT is uniformly fixed-parameter tractable (FPT) for
parameter $(t,\sb_t)$.  Theorem~\ref{the:main} states that there is a
uniform FPT-approximation algorithm for the
detection of strong $\TW_{\leq t}$-backdoor sets of size $k$, for
parameter $(t,k)$, as it is a fixed-parameter algorithm
that computes a solution that approximates the optimum with an error
bounded by a function of the parameter~\cite{Marx08b}.

\paragraph{Related work.}
Williams \etal~\cite{WilliamsGomesSelman03} introduced the notion of
backdoor sets and the
parameterized complexity of finding small backdoor sets was initiated by
Nishimura \etal~\cite{NishimuraRagdeSzeider04-informal}. They showed that
with respect to the classes of Horn formulas and of 2CNF formulas, the
detection of strong backdoor sets is fixed-parameter tractable. Their
algorithms exploit the fact that for these two base classes strong and
deletion backdoor sets coincide.  For other base classes, deleting
literals is a less powerful operation than applying partial truth
assignments. This is the case, for instance, for \RHorn, the class of renamable
Horn formulas. In fact, finding a deletion \RHorn-backdoor set is
fixed-parameter tractable~\cite{RazgonOSullivan09}, but it is open
whether this is the case for the detection of strong \RHorn-backdoor
sets. For clustering formulas, detection of
deletion backdoor sets is fixed-parameter tractable, detection of strong
backdoor sets is most probably not~\cite{NishimuraRagdeSzeider07}. Very
recently, the authors of the present paper showed
\cite{GaspersSzeider11a,GaspersSzeider12} that there are
FPT-approximation algorithms for the detection of strong
backdoor sets with respect to (i)~the base class of formulas with acyclic
incidence graphs, i.e., $\TW_{\le 1}$, and (ii)~the base class of nested formulas
(a proper subclass of $\TW_{\le 3}$ introduced by Knuth~\cite{Knuth90}). The present paper
generalizes this approach to base classes of bounded treewidth which
requires new ideas and significantly more involved combinatorial arguments.

We conclude this section by referring to a recent survey on the
parameterized complexity of backdoor sets~\cite{GaspersSzeider11b}.

\section{Preliminaries}
\label{section:prelims}

\paragraph{Graphs.}
Let $G$ be a simple, undirected, finite graph with vertex set $V=V(G)$ and edge set $E=E(G)$.
Let $S \subseteq V$\longversion{ be a subset of its vertices} and $v\in V$\longversion{ be a vertex}.
We denote by $G - S$ the graph obtained from $G$ by removing all vertices in $S$ and all edges incident to vertices in $S$.
We denote by $G[S]$ the graph $G - (V\setminus S)$.
The \emph{(open) neighborhood} of $v$ in $G$ is $N_G(v) = \set{u\in V : uv\in E}$, the \emph{(open) neighborhood} of $S$ in $G$ is $N_G(S) = \bigcup_{u\in S}N_G(u)\setminus S$, and their \emph{closed
neighborhoods} are $N_G[v] = N_G(v)\cup \set{v}$ and $N_G[S] = N_G(S)\cup S$, respectively. Subscripts may be omitted if the graph is clear from the context.
%

A \emph{tree decomposition} of $G$ is a pair 
$(\{X_i : i\in I\},T)$
where $X_i \subseteq V$, $i\in I$, and $T$ is a tree with elements
of $I$ as nodes
such that:
\longshort{\begin{enumerate}
  \item }{(1)} for each edge $uv\in E$, there is an $i\in I$ such that $\{u,v\} 
\subseteq X_i$, and
\longshort{\item}{(2)} for each vertex $v\in V$, $T[\set{i\in I: v\in X_i}]$ is a (connected) tree with at least one node.%
\longversion{\end{enumerate}}
The \emph{width} of a tree decomposition is $\max_{i \in I} |X_i|-1$.
The \emph{treewidth} \cite{RobertsonSeymour86} of $G$
is the minimum width taken over all tree decompositions
of $G$ and it is denoted by $\tw(G)$.

For other standard graph-theoretic notions not defined here, we refer to \cite{Diestel10}.

\paragraph{CNF formulas and satisfiability.}
We consider propositional formulas in conjunctive normal form (CNF) where no clause contains
a complementary pair of literals.
For a clause $c$, we write $\lit(c)$ and $\var(c)$ for the sets of literals and variables
occurring in $c$, respectively.
For a CNF formula $F$ we write $\cla(F)$ for its set of clauses,
$\lit(F) = \bigcup_{c\in \cla(F)} \lit(c)$ for its set of literals, and
$\var(F) = \bigcup_{c\in \cla(F)} \var(c)$ for its set of variables.
The \emph{size} of $F$ is $|F|=|\var(F)| + \sum_{c\in \cla(F)} (1+|\lit(c)|)$.

For a set $X\subseteq \var(F)$ we denote by $2^X$ the set of
all mappings $\tau:X\rightarrow \set{0,1}$, the \emph{truth assignments} on $X$.
A truth assignment $\tau\in 2^X$
can be extended to 
the literals over $X$ 
by setting $\tau(\neg x) = 1-\tau(x)$ for all $x\in X$.
The formula
$F[\tau]$ is obtained from $F$ by removing all clauses $c$
such that $\tau$ sets a literal of $c$ to~1, and removing the literals set to~0
from all remaining clauses.

A CNF formula $F$ is \emph{satisfiable} if there is some $\tau\in
2^{\var(F)}$ with $\cla(F[\tau])=\emptyset$. 
SAT is the $\NP$-complete problem of deciding whether a given CNF formula is
satisfiable~\cite{Cook71,Levin73}. \#SAT is the \#P-complete problem of
determining the number of distinct $\tau\in 2^{\var(F)}$ with $\cla(F[\tau])=\emptyset$ \cite{Valiant79b}.

\paragraph{Formulas with bounded incidence treewidth.}
The \emph{incidence graph} of a CNF formula $F$ is the bipartite graph $\inc(F)=(V,E)$ with
$V = \var(F) \cup \cla(F)$ and for a variable $x \in \var(F)$ and a clause $c \in \cla(F)$
we have $x c \in E$ if $x\in \var(c)$. The \emph{sign} of the edge $x c$ is \emph{positive}
if $x\in \lit(c)$ and \emph{negative} if $\neg x \in \lit(c)$. Note that $|V|+|E|=|F|$.

The class \Treewidth contains all CNF formulas $F$ with $\tw(\inc(F))\le t$.
For any fixed $t\ge 0$ and any CNF formula $F\in \Treewidth$, a tree decomposition of $\inc(F)$ of width at most $t$ can be found by Bodlaender's algorithm \cite{Bodlaender96}
in time $O(|F|)$.
Given a tree decomposition of width at most $t$ of $\inc(F)$, the number of satisfying assignments of
$F$ can be determined in time $O(|F|)$ \cite{FischerMakowskyRavve06,SamerSzeider10}.

Finally, note that, if $\tau \in 2^X$ is a partial truth assignment for a CNF formula $F$, then $\inc(F[\tau])$ is an induced subgraph of $\inc(F)$,
namely $\inc(F[\tau])$ is obtained from $\inc(F) - X$ by removing each vertex corresponding to a clause that contains a literal $\ell$ with $\tau(\ell)=1$.

\paragraph{Backdoors.}
Backdoor sets are defined with respect to a fixed class $\cC$ of CNF
formulas, the \emph{base class}.
Let $F$ be a CNF formula and $B\subseteq \var(F)$.
$B$ is a \emph{strong} \emph{$\cC$-\BDS{}} of $F$ if $F[\tau]\in \cC$ for each $\tau \in 2^B$.
$B$ is a \emph{deletion $\cC$-\BDS{}} of $F$ if $F - B \in \cC$, where $F - B$ 
is obtained from $F$ by removing all literals in $\set{x, \neg x: x\in B}$ from its clauses.

If we are given a strong $\cC$-\BDS of $F$ of
size $k$, we can reduce the satisfiability of $F$ to the satisfiability
of $2^k$ formulas in $\cC$.
If $\cC$ is clause-induced (\ie, $F\in \cC$ implies $F'\in \cC$ for every CNF formula $F'$ with $\cla(F')\subseteq \cla(F)$),
any deletion $\cC$-\BDS of $F$
is a strong $\cC$-\BDS of $F$.
The interest in deletion \BDSs is motivated for base classes where they
are easier to detect than strong \BDSs.
The challenging problem is to find a strong 
or deletion
$\cC$-\BDS of size at most $k$ if it exists.
Denote by $\sb_t(F)$ the size of a smallest \STBDS.


\paragraph{Graph minors.}
The operation of \emph{merging} a subgraph $H$ or a vertex subset $V(H)$ of a graph $G$ into a vertex $v$
produces the graph $G'$ such that $G'-\set{v}=G- V(H)$ and $N_{G'}(v) = N_G(H)$.
The \emph{contraction} operation merges a connected subgraph.
The \emph{dissolution} operation contracts an edge incident to a vertex of degree~2.

A graph $H$ is a \emph{minor} of a graph $G$ if $H$ can be obtained from a subgraph of $G$
by contractions.
If $H$ is a minor of $G$, then one can find a model of $H$ in $G$.
A \emph{model} of $H$ in $G$ is a set of vertex-disjoint connected subgraphs
of $G$, one subgraph $C_u$ for each vertex $u$ of $H$, such that if $uv$ is an edge in $H$, then
there is an edge in $G$ with one endpoint in $C_u$ and the other in $C_v$.

A graph $H$ is a \emph{topological minor} of a graph $G$ if $H$ can be obtained from a subgraph of $G$
by dissolutions.
If $H$ is a topological minor of $G$, then $G$ has a topological model of $H$.
A \emph{topological model} of $H$ in $G$ is a subgraph of $G$ that can be obtained
from $H$ by replacing its edges by independent paths. A set of paths is \emph{independent} if
none of them contains an interior vertex of another. We also say that $G$ contains a \emph{subdivision} of
$H$ as a subgraph.

\paragraph{Obstructions to small treewidth.}
It is well-known that $\tw(G) \ge \tw(H)$ if $H$ is a minor of $G$.
We will use the following three (classes of) graphs to lower bound the treewidth of a graph containing any of them as a minor.
See Figure \ref{fig:wall}.
The complete graph $K_r$ has treewidth $r-1$.
The complete bipartite graph $K_{r,r}$ has treewidth $r$.
The \emph{$r$-wall} is the graph
$W_r=(V,E)$ with vertex set $V = \{(i, j) : 1 \le i \le r, 1 \le j \le r\}$ in which two vertices
$(i,j)$ and $(i',j')$ are adjacent \myiff either $j'=j$ and $i' \in \{i-1, i+1\}$, or $i'=i$ and $j'=j+(-1)^{i+j}$.
We say that a vertex $(i,j)\in V$ has horizontal index $i$ and vertical index $j$.
The $r$-wall has treewidth at least $\lfloor \frac{r}{2} \rfloor$ (it is a minor of the $\lfloor \frac{r}{2} \rfloor \times \lfloor \frac{r}{2} \rfloor$-grid,
which has treewidth $\lfloor \frac{r}{2} \rfloor$ if $\lfloor \frac{r}{2} \rfloor\ge 2$ \cite{RobertsonSeymour91}).

We will also need to find a large wall as a topological minor if the formula has large incidence treewidth.
Its existence is guaranteed by a theorem of Robertson and Seymour.

\begin{theorem}[\cite{RobertsonSeymour86b}]
 For every positive integer $r$, there exists a constant $f(r)$ such that if a graph $G$
 has treewidth at least $f(r)$, then $G$ contains an $r$-wall as a topological minor.
\end{theorem}

\noindent
By \cite{RobertsonSeymourThomas94}, $f(r) \le 20^{64 r^5}$.
For any fixed $r$,
we can use the cubic algorithm by Grohe \etal~\cite{GroheKMW11}
to find a topological model of an $r$-wall in a graph $G$ if $G$ contains an $r$-wall as a topological minor.

\tikzset{var/.style={inner sep=.15em,circle,fill=black,draw},
         clause/.style={minimum size=1mm,rectangle,fill=white,draw},
         label distance=-2pt}

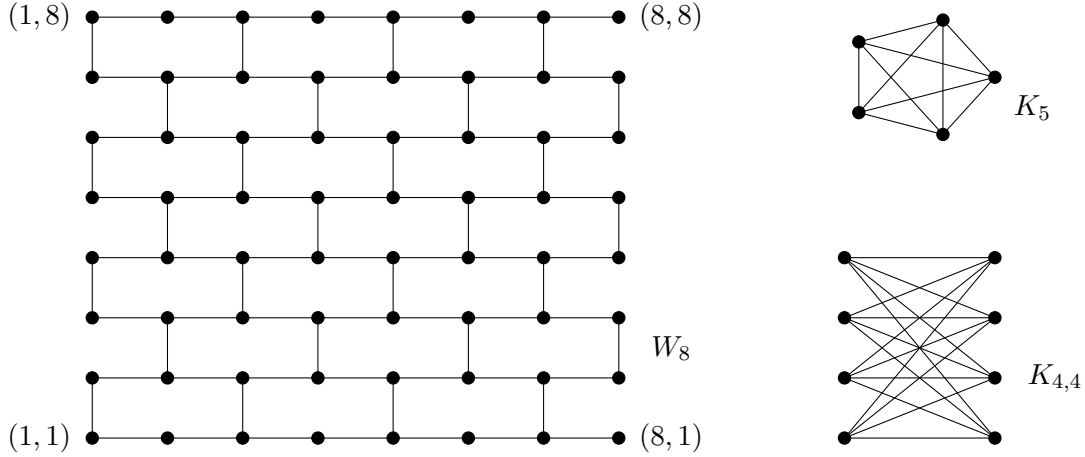
\begin{figure}[tb]
 \centering
  \begin{tikzpicture}[xscale=1,yscale=0.8]
   \pgfmathtruncatemacro\dimw{8}
   \foreach \x in {1,2,...,\dimw}
    \foreach \y in {1,2,...,\dimw} {
     \node at (\x,\y) [var] {};
     \ifnum\x<\dimw
       \draw (\x,\y) -- +(1,0);
     \fi
     \pgfmathtruncatemacro\even{round(mod(\x+\y,2))}
     \let\res\even
     \ifnum\res=0
       \ifnum\y<\dimw
         \draw (\x,\y) -- +(0,1);
       \fi
     \fi
    }

  \node[xshift=0.7cm] at (\dimw,2.5) {$W_{\dimw}$};
  \node at (0.3,1) {$(1,1)$};
  \node at (0.3,\dimw) {$(1,\dimw)$};
  \node[xshift=0.7cm] at (\dimw,1) {$(\dimw,1)$};
  \node[xshift=0.7cm] at (\dimw,\dimw) {$(\dimw,\dimw)$};
  

  \begin{scope}[xshift=12cm,yshift=7cm]
   \pgfmathtruncatemacro\dimc{5}
   \foreach \x in {1,2,...,\dimc} {
    \pgfmathtruncatemacro\deg{round(\x*360/\dimc)}
    \node at +(\deg:1cm) [var] {};
    \foreach \y in {\x,...,\dimc} {
     \pgfmathtruncatemacro\degy{round(\y*360/\dimc)}
     \draw +(\deg:1cm) -- +(\degy:1cm);
    }
  }

  \node at (1.5,-0.5) {$K_{\dimc}$};
  \end{scope}

  \begin{scope}[xshift=11cm,yshift=0cm]
   \pgfmathtruncatemacro\dimcb{4}
   \foreach \x in {1,2,...,\dimcb} {
    \node at (0,\x) [var] {};
    \node at (2,\x) [var] {};
    \foreach \y in {1,...,\dimcb} {
     \draw (0,\x) -- (2,\y);
    }
  }

  \node at (2.8,2) {$K_{\dimcb,\dimcb}$};
  \end{scope}

  \end{tikzpicture}
  \caption{Some graphs with treewidth $4$.}
  \label{fig:wall}
\end{figure}


\section{The algorithms}

We start with the overall outline of our algorithms.
We rely on the following two lemmas whose proofs we defer
to the next two subsections.

\begin{lemma}\label{lem:wall}
  There is a quadratic-time algorithm that, given a CNF formula~$F$, two
  constants~$t\ge 0$, $k\ge 1$, and a topological model of a $\kwall$-wall in $\inc(F)$,
  computes a set $S^*\subseteq \var(F)$ of constant size such
  that every \STBDS of size at most $k$ contains a variable from $S^*$.
\end{lemma}

\begin{lemma}\label{lem:mso}
  There is a linear-time algorithm that,
  given a CNF formula~$F$, a constant $t\ge 0$, and a tree decomposition of
  $\inc(F)$ of constant width, computes a smallest
  \STBDS of $F$.
\end{lemma}

\noindent
Lemma \ref{lem:mso} will be invoked with a tree decomposition of $\inc(F)$
of width at most $\ktw$.
The functions $\kwall$ and $\ktw$ are related by the bound
from \cite{RobertsonSeymourThomas94}, implying that every graph
either has treewidth at most $\ktw$, or it has a $\kwall$-wall as a topological minor.
Here,
\begin{align*}
 \ktw &:= 20^{64 \cdot (\kwall)^5},\\
 \kwall &:= (2t+2) \cdot (1+\sqrt{\kobs}),\\
 \kobs &:= 2^{k} \cdot \ksame + k,\\
 \ksame &:= 3 (\tnb)^2 t 2^{2k},\text{ and}\\
 \tnb &:= \lceil 16 (t+2) \log (t+2) \rceil.
\end{align*}
The other functions of $k$ and $t$ will be used in Subsection \ref{subsec:wall}.

Theorem~\ref{the:main} can now be proved as follows.

\begin{proof}[Proof of Theorem~\ref{the:main}]
  Let
  $t,k\geq 0$ be constants, let $F$ be the given CNF formula, with
  $\Card{F}=n$ and let $G:=\inc(F)$. Using Bodlaender's
  algorithm \cite{Bodlaender96} we can decide in linear time whether $\tw(G)\leq \ktw$,
  and if so, compute a tree decomposition of smallest width in
  linear time. If indeed $\tw(G)\leq \ktw$, we use Lemma~\ref{lem:mso}
  to find a smallest strong $\TW_{\leq t}$\hy backdoor set $B$ of
  $F$. If $\Card{B}\leq k$ we output $B$, otherwise we output~NO.
  
  If $\tw(G)> \ktw$ then we proceed as follows.
  If $k=0$, we output NO. Otherwise, by
  \cite{RobertsonSeymourThomas94}, $G$ has a $\kwall$-wall as a
  topological minor, and by means of Grohe \etal's algorithm~\cite{GroheKMW11}, we can
  compute a topological model of a $\kwall$\hy wall in $G$ in time
  $O(n^3)$. By Lemma~\ref{lem:wall}, we can find in time $O(n^2)$ a set
  $S^*\subseteq \var(F)$ of constant size such that every \STBDS of $F$ of size at
  most $k$ contains a variable from $S^*$. For each $x\in S^*$, the
  algorithm recurses on both formulas $F[x=0]$ and $F[x=1]$ with parameter $k-1$.  If both
  recursive calls return \STBDSs $B_{\neg x}$ and $B_{x}$, then $\{x\}
  \cup B_{x} \cup B_{\neg x}$ is a \STBDS of $F$. We can upper bound its size
  $s(k)$ by the recurrence
   $s(k) \le 1+2\cdot s(k-1)$,
  with $s(0)=0$ and $s(1)=1$. The recurrence is satisfied by setting $s(k) = 2^k-1$.
  In case a recursive call returns NO,
  no \STBDS of $F$ of size at most~$k$ contains $x$.
  Thus, if for some $x\in S^*$, both recursive calls return backdoor sets,
  we obtain a backdoor set of $F$ of size at most $2^k-1$, and if
  for every $x\in S^*$, some recursive call returns NO, $F$ has no
  \STBDS of size at most $k$.
  
  The number of nodes of the search tree modeling the recursive calls of this algorithm
  is a function of $k$ and $t$ only (and therefore constant), and in each node, the time spent by the
  algorithms is $O(n^2)$.
  The overall running time is thus dominated by
  the cubic running time of Grohe \etal's algorithm, hence we
  arrive at a total running time of $O(n^3)$.
\end{proof}

\noindent
Theorem~\ref{the:counting} follows easily from Theorem~\ref{the:main},
by computing first a backdoor set and evaluating the number of
satisfying assignments for all reduced formulas.
We present an alternative proof that does not rely on Lemma \ref{lem:mso}.
Instead of computing a backdoor set,
one can immediately compute the number of satisfying assignments of $F$ by dynamic programming
if $\tw(\inc(F)) \le \ktw$.

\begin{proof}[Proof of Theorem~\ref{the:counting}]
  Let $k,t\ge 0$ be two integers and assume we are given a CNF formula $F$
  with $\Card{F}=n$ and $\sb_t(F)\le k$.  We will compute the
  number of satisfying truth assignments of $F$, denoted $\#(F)$.  As
  before we use Bodlaender's linear-time algorithm \cite{Bodlaender96} to decide whether $\tw(G)\leq
  \ktw$, and if so, to compute a tree decomposition of smallest
  width. If $\tw(G)\leq \ktw$ then we use the tree decomposition and\longversion{, for instance,} the
  algorithm of \cite{SamerSzeider10} to compute $\#(F)$ in time
  $O(n)$.
  
  If $\tw(G)> \ktw$ then we compute, as in the proof of
  Theorem~\ref{the:main}, a strong $\TW_{\leq t}$\hy backdoor set $B$ of
  $F$ of size at most $2^k$ in time $O(n^3)$. For each
  $\tau\in 2^B$ the formula $F[\tau]$ belongs to
  $\TW_{\leq t}$. Hence we can compute $\#(F[\tau])$ in time $O(n)$
  by first computing a tree decomposition of width at most
  $t$, and then applying the counting algorithm
  of~\cite{SamerSzeider10}.  We obtain $\#(F)$ by taking
  $\sum_{\tau\in 2^B} 2^{d(F,\tau)} \,\#(F[\tau])$ where
  $d(F,\tau)=\Card{\var(F)\setminus (B\cup \var(F[\tau]))}$ denotes the number
  of variables that disappear from~$F[\tau]$ without being instantiated.
\end{proof}

\subsection{The incidence graph has a large wall as a topological minor}
\label{subsec:wall}

This subsection is devoted to the proof of Lemma \ref{lem:wall}
and contains the main combinatorial arguments of this paper.
Let $G=(V,E)=\inc(F)$ and suppose we are given
a topological model of a $\kwall$-wall in~$G$.
We start with the description of the algorithm.

A \emph{\Wobs{}} is a subgraph of $G$ that is a subdivision of a $(2t+2)$-wall.
Since a \Wobs, and any graph having a \Wobs as a topological minor, has treewidth at least $t+1$,
we have that for each assignment to the variables of a 
\STBDS, at least one vertex from each \Wobs vanishes in the incidence graph of the reduced formula.
Using the $\kwall$-wall, we now find a set $\cO$ of $\kobs$ vertex-disjoint \Wobss in $G$.

\begin{lemma}\label{lem:Wobs}
 Given a topological model of a $\kwall$-wall in $G$, a set of $\kobs$ vertex-disjoint \Wobss can be found in linear time.
\end{lemma}
\begin{proof}
For any two integers $i$ and $j$ with $1\le i,j\le \kwall/(2t+2)$,
the subgraph of a $\kwall$-wall induced on all vertices $(x,y)$
with $(i-1)\cdot (2t+2)+1 \le x \le i\cdot (2t+2)$
and $(j-1)\cdot (2t+2)+1 \le y \le j\cdot (2t+2)$
is a $(2t+2)$-wall.
A corresponding \Wobs can be found in $G$ by replacing edges by the independent paths they model
in the given topological model.
The number of \Wobss defined this way is $\left\lfloor \frac{\kwall}{2t+2} \right\rfloor^2 \ge \left( \frac{\kwall}{2t+2}-1 \right)^2 \ge \kobs$.
\end{proof}

\noindent
Denote by $\cO$ a set of $\kobs$ vertex-disjoint \Wobss obtained via Lemma \ref{lem:Wobs}.
A backdoor variable can destroy a \Wobs either because it participates in the \Wobs, or because
every setting of the variable satisfies a clause that participates in the \Wobs.

\begin{definition}
Let $x$ be a variable and $W$ a \Wobs in $G$.
We say that $x$ \emph{kills} $W$ if neither $\inc(F[x = \true])$ nor $\inc(F[x = \false])$ contains $W$ as a subgraph.
We say that $x$ kills $W$ \emph{internally} if $x\in V(W)$, and that $x$ kills $W$ \emph{externally} if $x$ kills $W$ but does not kill it internally.
In the latter case, $W$ contains a clause $c$ containing $x$ and a clause $c'$ containing $\neg x$ and we say that
$x$ kills $W$ (externally) \emph{in} $c$ and $c'$.
\end{definition}

\noindent
Our algorithm will perform a series of $3$ nondeterministic steps to guess some properties about the \STBDS it searches.
Each such guess is made out of a number of choices that is upper bounded by a function of $k$ and $t$.
At any stage of the algorithm, a \emph{valid} \STBDS is one that satisfies all the properties that have been guessed.
For a fixed series of guesses, the algorithm will compute a set $S\subseteq \var(F)$ such that every valid \STBDS of size at most $k$ contains a variable from $S$.
To make the algorithm deterministic, execute each possible combination of nondeterministic steps. The union of all
$S$, taken over all combinations of nondeterministic steps, forms a set $S^*$ and each \STBDS of size at most $k$ contains a variable from $S^*$.
Bounding the size of each $S$ by a function of $k$ and $t$ enables us to bound $|S^*|$ by a function of $k$ and $t$, and this will prove the lemma.

For any \STBDS of size at most $k$, at most $k$ \Wobss from $\cO$ are killed internally since they are vertex-disjoint.
The algorithm guesses $k$ \Wobss from $\cO$ that may be killed internally.
Let $\cO'$ denote the set of the remaining \Wobss, which need to be killed externally by any valid \STBDS.

Suppose $F$ has a valid \STBDS $B$ of size $k$. Then, $B$ defines a partition of $\cO'$ into
$2^k$ parts where for each part, the \Wobss contained in this part are killed externally by the same set of variables from $B$.
Since $|\cO'| = \kobs - k = 2^k \cdot \ksame$, at least one of these parts contains at least $\ksame$ \Wobss from $\cO'$.
The algorithm guesses a subset $\cO_s \subseteq \cO'$ of $\ksame$ \Wobs from this part and it guesses how many variables from the
\STBDS kill the \Wobss in this part externally.

Suppose each \Wobs in $\cO_s$ is killed externally by the same set of~$\ell$ backdoor variables, and no other backdoor variable
kills any \Wobs from $\cO_s$. Clearly, $1\le \ell \le k$.
Compute the set of external killers for each \Wobs in $\cO_s$. Denote by $Z$ the common external killers of the \Wobs in $\cO_s$.
The presumed \BDS contains exactly $\ell$ variables from $Z$ and no other variable from the \BDS kills any \Wobs from~$\cO_s$.

We will define three rules for the construction of $S$, and the algorithm will execute the first applicable rule.

\begin{myrule}[Few Common Killers]\label{rule:fewkillers}
 If $|Z|\le 6 k \tnb$, then set $S:=Z$.
\end{myrule}

\noindent
Before being able to state the other two rules, we come to the central combinatorial object in this paper.
For each \Wobs $W\in \cO_s$, we compute a valid \obstemp. An \emph{\obstemp{}} $\OT(W)$ of a \Wobs $W\in \cO_s$ is a
triple $(\cB(W),P,R)$, where
\begin{itemize}
 \item $\cB(W)$ is a bipartite graph whose vertex set is bipartitioned into the two independent sets $Z$ and $Q_W$, where $Q_W$ is a set of new vertices,
 \item $P$ is a partition of $V(W)$ into \emph{regions} such that for each region $A\in P$, we have that $W[A]$ is connected, and
 \item $R: Q_W\rightarrow P$ is a function associating a region of $P$ with each vertex in $Q_W$.
\end{itemize}
An \obstemp $\OT(W)=(\cB(W),P,R)$ of a \Wobs $W\in \cO_s$ is \emph{valid} if
it satisfies the following properties:
\begin{description}
 \item[(1) only existing edges:] for each $q\in Q_W$, $N_{\cB(W)}(q) \subseteq N_G(R(q))$,
 \item[(2) private neighbor:] for each $q\in Q_W$, there is a $z\in N_{\cB(W)}(q)$, called $q$'s \emph{private neighbor}, such that there is no other $q'\in N_{\cB(W)}(z)$ with $R(q')=R(q)$,
 \item[(3) degree-$Z$:] for each $z\in Z$, $d_{\cB(W)}(z) \ge 1$,
 \item[(4) degree-$Q_W$:] for each $q\in Q_W$, $\tnb \le d_{\cB(W)}(q) \le 3 \tnb$, and
 \item[(5) vulnerable vertex:] for each $q\in Q_W$, there is at most one vertex $v\in R(q)$, called $q$'s \emph{vulnerable vertex}, such that $N_G(v) \cap Z \not \subseteq N_{\cB(W)}(q)$.
\end{description}
We will use the \obstemps to identify a set of vertices that has a non-empty intersection with every valid \STBDS of size $k$.
Intuitively, an \obstemp chops up the vertex set of a \Wobs into regions. We will suppose the existence of a valid \BDS $B$ of size $k$ avoiding
a certain bounded set of variables and derive a contradiction using the \obstemps.
This is done by showing that for at least one $\tau \in 2^B$,
many regions remain in $F[\tau]$, so that we can contract each of them and construct a treewidth
obstruction using the contracted vertices. Each vertex from $Q_W$ models a contraction of a region, and its neighborhood
models a potential set of common external killers neighboring the contracted region. This explains Property (1).
Property (2) becomes handy when a region has many vertices from $Q_W$ that are associated with it. Namely, when we
contract regions, we would like to be able to guarantee a lower bound on the number of edges of the resulting graph
in terms of $|Q_W|$. To ensure that this lower bound translates into a lower bound in terms of $|Z|$, we need Property (3).
The degree lower bound of the next property will be needed so we can patch together a treewidth obstruction out of the
pieces modeled by the vertices in $Q_W$. The upper bound on the degree is required to guarantee that sufficiently
many vertices from $Q_W$ are not neighboring $B$.
Finally, the last property will be used to guarantee that for every $q\in Q_W$, if $B \cap N_{\cB(W)}(q)=\emptyset$, then
there is a truth assignment $\tau \in 2^B$ such that no vertex from $q$'s region is removed from $\inc(F)$ by
applying $\tau$ (see Lemma~\ref{lem:vulnerable}).

In the following lemma, we give a procedure to compute valid \obstemps.

\begin{lemma}\label{lem:obstemp}
 For each \Wobs $W\in \cO_s$, a valid \obstemp can be computed in time $O(|V(W)|^2 +|V(W)|\cdot |Z|)$.
\end{lemma}
\begin{proof}
 We describe a procedure to compute a valid \obstemp $(\cB(W),P,R)$.
 It starts with $Q_W$ initially empty.
 Compute an arbitrary rooted spanning tree $T$ of $W$.
 For a node $v$ from $T$, denote by $T_v$ the subtree of $T$ rooted at $v$.
 The set of children in $T$ of a node $v$ is denoted $C_T(v)$ and its parent $p_T(v)$.
 For a subforest $T'\subseteq T$, denote by $Z(T') = Z \cap N_G(T')$ the subset of vertices from $Z$ that have a neighbor from $T'$ in $G$.
 The \emph{weight} $w(T')$ of $T'$ is $|Z(T')|$.
 We denote by $B_v = Z(T_v) \setminus Z(T_v - \{v\})$ the vertices from $Z$ that are incident to $v$ in $G$ but to no other node
 from $T_v$.
 If $uv\in E(T)$, then denote by $T_u(uv)$ the subtree obtained from $T$ by removing all nodes that are closer to $v$ than to $u$ in $T$
 (removing the edge $uv$ decomposes $T$ into $T_u(uv)$ and $T_v(uv)$).

 \begin{description}
  \item[(A)] If $w(T)>3 \tnb$, then select a new root $r(T)$ of $T$ such that for every child $c \in C_T(r(T))$ of $r(T)$ we have that $w(T - V(T_c)) \ge \tnb$.
 
  \item[(B)] Select a node $v$ in $T$ as follows.
 If $w(T)\le 3 \tnb$, then set $v:=r(T)$. Otherwise, select the node $v$ at maximum depth in $T$ such that $w(T_v) \ge \tnb$.
 The vertices from $T_v$ will constitute one region $A_v$ of $P$.
 Denoting $s= 3 \tnb - w(T_v-\{v\})$,
 we will now add a set of $\left\lceil \frac{|B_v|}{s} \right\rceil$ vertices to $Q_W$. All of them are associated with the region $A_v$.
 Denote these new vertices $q_1, \dots, q_{\lceil |B_v|/s \rceil}$, and denote the vertices in $B_v$ by $b_1, \dots, b_{|B_v|}$.
 For each $i, 1\le i\le \lceil |B_v|/s \rceil$, we set $N(q_i):= Z(T_v-\{v\}) \cup \{b_{(i-1)\cdot s+1}, \dots, b_{i\cdot s}\}$; indices are taken modulo $|B_v|$.
 If $v \neq r(T)$, then set $T:=T_{p(v)}(v p(v))$ (\ie, remove $V(T_v)$ from $T$) and go to Step (A).
 \end{description}

\noindent
 Now, we prove that this procedure computes a valid \obstemp.

 First we show that in case $w(T)>3 \tnb$, Step (A) is able to find a root $r(T)$ such that there is no $c \in C_T(r(T))$ with $w(T - V(T_c)) < \tnb$.
 Suppose that $T$ has no node $u$ such that $w(T_u(uv)) \ge \tnb$ for every $v\in N_T(u)$.
 Then, there is an infinite sequence of nodes $u_1, u_2, \dots$ such that $u_i$ neighbors $u_{i+1}$ and
 $w(T_{u_i}(u_i u_{i+1}))<\tnb$.
 Let $j$ be the smallest integer such that $u_i=u_j$ for some integer $i$ with $1\le i < j$.
 Since $T$ is acyclic, we have that $i+2=j$.
 But then, $w(T) \le w(T_{u_i}(u_i u_{i+1})) + w(T_{u_{i+1}}(u_i u_{i+1})) \le 2 \tnb -2$, contradicting the assumption that $w(T)>3 \tnb$.

 We observe that all edges of $\cB(W)$ have one endpoint in $Z$ and the other in $Q_W$. Thus, $Z \uplus Q_W$ is a bipartition of $\cB(W)$ into independent sets.
 The set $V(W)$ is partitioned into disjoint connected regions since each execution of Step (B) defines a new region equal to the vertices of a subtree of $T$, which is initially a spanning tree of $W$ and
 is removed from $T$ at the end of Step (B).
 
 Consider one execution of Step (B) of the procedure above. We will show that Properties (1)--(5) of a valid \obstemp
 hold for the relevant vertices considered in this execution, and this will guarantee these properties for all vertices.
 Property (1) is ensured for all new vertices introduced in $Q_W$ since $B_v\cup Z(T_v - \{v\}) = Z(T_v) \subseteq N_G(A_v)$.
 The private neighbor of a vertex $q_i$ is $b_{i\cdot s}$ if $i<\lceil |B_v|/s \rceil$ and $b_{|B_v|}$ if $i=\lceil |B_v|/s \rceil$.
 Property (3) is ensured for all vertices in $Z\cap N_G(T_v)$ since all of them receive at least one new neighbor in $\cB(W)$.
 For the lower bound of Property (4), we first show that at any time during the execution of this procedure, either $T$ is empty or $w(T) \ge \tnb$.
 Initially, this is true since $|Z|\ge \tnb$ (by Rule \ref{rule:fewkillers}) and every vertex from $Z$ is an external killer of $W$.
 This remains true since Step (A) makes sure that whenever the vertex $v$ chosen by Step (B) is not the root of $T$, 
 $w(T - V(T_v)) \ge \tnb$. Thus, Step (B) always finds a node $v$ such that $w(T_v) \ge \tnb$.
 Therefore, every vertex that is added to $Q_W$ has at least $\tnb$ neighbors.
 For the upper bound of Property (4), observe that since $W$ has maximum degree at most $3$, $T$ also has maximum degree at most $3$.
 Thus, $w(T_v - \{v\}) \le 3 (\tnb-1)$ since each tree in $T_v - \{v\}$ has weight at most $\tnb-1$ by the selection of $v$.
 Therefore, $d_{\cB(W)}(q_i) \le 3 \tnb$.
 Property (5) holds since $Z(T_v - \set{v}) \subseteq N_{\cB(W)}(q_i)$ and thus $v$ is the only vertex that can be vulnerable for~$q_i$.

 \medskip

 We upper bound the running time of the procedure as follows.
 A spanning tree of $W$ can be computed in time $O(|V(W)|)$.
 In a bottom-up fashion starting at the leaves of $T$, one can precompute $Z(T_u(uv))$ for all $uv\in E(T)$ in time $O(|V(W)|\cdot |Z|)$.
 Then, Step (A) can be implemented such that each execution runs in time $O(|V(W)|)$. One execution of Step (B) takes time $O(|V(W)|+|Z|)$.
 Since Steps (A) and (B) are executed at most $|V(W)|$ times, the running time is $O(|V(W)|\cdot (|V(W)|+|Z|))$.
\end{proof}

\noindent
The bipartite graph $\cB_m(\cO_s)$ is obtained by taking the union of all $\cB(W)$, $W\in \cO_s$.
Its subgraphs $\cB(W)$, $W\in \cO_s$, share the same vertex subset $Z$ but
the vertex subsets $Q_W$, $W\in \cO_s$, are pairwise disjoint.
The vertex set of $\cB_m(\cO_s)$
is $Z \uplus Q_m$, where $Q_m = \bigcup_{W\in \cO_s} Q_W$.

\begin{myrule}[Multiple Neighborhoods]\label{rule:MultiNb}
 If there is a subset $L \subseteq Z$ such that $L$ is the neighborhood of at least $t\cdot 2^k+1$ vertices in $\cB_m(\cO_s)$,
 then set $S:=L$.
\end{myrule}

\noindent
Obtain a bipartite graph $\cB(\cO_s)$ from $\cB_m(\cO_s)$ by repeatedly and exhaustively deleting vertices from $Q_m$ whose neighborhood
equals the neighborhood of some other vertex from $Q_m$. Denote the vertex set of the resulting graph $\cB(\cO_s)$ by $Z \uplus Q$.

\begin{myrule}[No Multiple Neighborhoods]\label{rule:NoMultiNb}
 Set $S$ to be the $6 k \tnb$ vertices from $Z$ of highest degree in $\cB(\cO_s)$ (ties are broken arbitrarily).
\end{myrule}

\noindent
This finishes the description of the algorithm.
The correctness of Rule \ref{rule:fewkillers} is obvious since any valid \STBDS\
contains $\ell$ variables from $Z$ and $\ell\ge 1$.
To prove the correctness of Rules \ref{rule:MultiNb} and \ref{rule:NoMultiNb},
we need the following lemma.

\begin{lemma}\label{lem:vulnerable}
 Let $W\in \cO_s$ be a \Wobs, $\OT(W)$ be a valid \obstemp of $W$ and $q\in Q_W$.
 Let $B$ be a valid \STBDS such that $B \subseteq \var(F) \setminus N_{\cB(W)}(q)$.
 There is a truth assignment $\tau$ to $B$ such that $\inc(F[\tau])$ contains all vertices from $R(q)$.
\end{lemma}
\begin{proof}
 Since $B$ is valid, it contains no variable from $R(q) \subseteq V(W)$.
 Thus, $\inc(F[\tau])$ contains all variables from $R(q)$.
 A truth assignment $\tau$ removes a clause $c\in R(q)$ from the incidence graph \myiff $c$ contains a literal $l$ such that $\tau(l)=1$.
 We show that no variable from $B$ appears both positively and negatively in the clauses from $R(q)$, and therefore there is a truth
 assignment $\tau$ to $B$ such that $\inc(F[\tau])$ contains all vertices from $R(q)$.

 Assume, for the sake of contradiction, that there is a variable $b\in B$ such that $b\in \lit(c)$ and $\neg b \in \lit(c')$ and $c,c'\in R(q)$.
 We have that $b \in Z$ because $b$ is in a valid \STBDS and $b$ is an external killer of $W$.
 Since $b \in (N_G(R(q)) \cap Z) \setminus N_{\cB(W)}(q)$, we conclude that $q$ has a vulnerable vertex $v$.
 But, since $v$ is the only vulnerable vertex of $q$, by Property (5), $c=c'$.
 We arrive at a contradiction, since no clause contains a variable and its negation.
\end{proof}

\begin{lemma}
 Rule \ref{rule:MultiNb} is sound.
\end{lemma}
\begin{proof}
 Let $Q_L \subseteq Q_m$ be a set of $t\cdot 2^k+1$ vertices such that for each $q\in Q_L$, $N_{\cB_m(\cO_s)}(q)=L$.
 For the sake of contradiction, suppose $F$ has a valid \STBDS $B$ of size~$k$ with $B\cap S = \emptyset$.
 By Lemma \ref{lem:vulnerable}, for each $q\in Q_L$, there is a truth assignment $\tau$ to $B$ such that
 $\inc(F[\tau])$ contains all vertices from $R(q)$.
 But there are at most $2^k$ truth assignments to $B$. Therefore, for at least one truth assignment $\tau$ to $B$,
 there is a set $Q_L' \subseteq Q_L$ of at least $\lceil |Q_L|/2^k \rceil = t+1$ vertices such that
 $\inc(F[\tau])$ contains all vertices from $R(q), q\in Q_L'$.
 By Property (2), no two distinct $q,q'\in Q_L$ are assigned to the same region.
 Consider the subgraph of $\inc(F[\tau])$ induced on all vertices in $L$ and $R_q, q\in Q_L'$.
 Contracting each region $R(q), q\in Q_L'$, one obtains
 a supergraph of a $K_{t+1,t+1}$. Thus, $\inc(F[\tau])$ has a $K_{t+1,t+1}$ as a minor,
 implying that its treewidth is at least $t+1$, a contradiction.
\end{proof}

\noindent
The correctness of Rule \ref{rule:NoMultiNb} will be shown with the use of a theorem by Mader.
\begin{theorem}[\cite{Mader68}]\label{thm:mader}
 Every graph $G=(V,E)$ with $|E| \ge c(x) \cdot |V|$ has a $K_x$-minor,
 where $c(x) = 8x\log x$.
\end{theorem}

\noindent
For large $x$, the function $c(x)$ can actually be improved to $c(x) = (\alpha + o(1)) x \sqrt{\log x}$ where $\alpha = 0.319\dots$ is an explicit constant,
and random graphs are extremal \cite{Thomason01}.

\begin{lemma}\label{lem:NoMultiNbSound}
 Rule \ref{rule:NoMultiNb} is sound.
\end{lemma}
\begin{proof}
Suppose $F$ has a valid \STBDS $B$ of size~$k$ with $B\cap S = \emptyset$.
To arrive at a contradiction, we exhibit a truth assignment $\tau$ to $B$ such that $\inc(F[\tau])$ has treewidth at least $t+1$.

\begin{claim}\label{claim:q}
 There is a truth assignment $\tau\in 2^B$ and a set $Q'\subseteq Q$ with $|Q'|\ge \frac{|Z| \cdot |\cO_s|}{3 \tnb t 2^{2k+1}}$
 such that $\inc(F[\tau])$ contains all vertices from $\bigcup_{q\in Q'} R(q)$.
\end{claim}
\noindent
To prove the claim, we first show a lower bound on $|Q \setminus N_{\cB(\cO_s)}(B)|$ in terms of $|Z|$ and $|\cO_s|$.

Since, by Property (3), each vertex $z\in Z$ has degree at least one in $\cB(W)$, $W\in \cO_s$, there are at least
$|Z| \cdot |\cO_s|$ edges in $\cB_m(\cO_s)$.
Since, by Property (4), each vertex from $Q_m$ has degree at most $3 \tnb$, we have that
$|Q_m| \ge \frac{|Z| \cdot |\cO_s|}{3 \tnb}$.
By Rule \ref{rule:MultiNb}, no set of $t\cdot 2^k+1$ vertices from $Q_m$ has the same neighborhood.
Therefore, $|Q| \ge \frac{|Z| \cdot |\cO_s|}{3 \tnb t 2^k}$.
Let $d$ denote the number of edges in $\cB(\cO_s)$ with one endpoint in $B$. Thus, $|N_{\cB(\cO_s)}(B)| \le d$.
Since $|S|\ge 6 |B| \tnb$ and the degree of any vertex in $S$ is at least the degree of any vertex in $B$, we have that
the number of edges incident to $S$ is at least $6 \tnb d$ in $\cB(\cO_s)$. Thus, $|Q| \ge \frac{6 \tnb d}{3 \tnb} = 2d$.
Therefore, $N_{\cB(\cO_s)}(B)$ contains at most half the vertices of $Q$, and we have that
$|Q \setminus N_{\cB(\cO_s)}(B)| \ge \frac{|Z| \cdot |\cO_s|}{3 \tnb t 2^{k+1}}$.

By Lemma \ref{lem:vulnerable}, for every $q\in Q \setminus N_{\cB(\cO_s)}(B)$ there is a truth assignment $\tau \in 2^B$
such that $\inc(F[\tau])$ contains all vertices from $R(q)$. Since $|2^B| = 2^k$, there is a truth assignment $\tau \in 2^B$
and a subset $Q' \subseteq Q \setminus N_{\cB(\cO_s)}(B)$ of at least $|Q \setminus N_{\cB(\cO_s)}(B)|/2^k \ge \frac{|Z| \cdot |\cO_s|}{3 \tnb t 2^{2k+1}}$ vertices
such that $\inc(F[\tau])$ contains all vertices from $R(q)$ for every $q\in Q'$.
%
%
%
This proves Claim \ref{claim:q}.\hfill $\lrcorner$
\bigskip

\noindent
Let $H':= \cB(\cO_s)[Z' \cup Q']$ where $Z':=Z\setminus B$ and $Q'$ is as in Claim \ref{claim:q}. Thus, no vertex from $Z'$ and no vertex from $\bigcup_{q\in Q'} R(q)$
is removed from the incidence graph by applying the truth assignment $\tau$ to $F$.
We will now merge vertices from $H'$ in such a way that we obtain a minor of $\inc(F[\tau])$.
To achieve this, we repeatedly merge a part $A\in P$ into a vertex $z\in Z$ such that $z$ has a neighbor $q$ in $H'$ such that $R(q)=A$.
In the incidence graph, this corresponds to contracting $R(q) \cup \{z\}$ into the vertex $z$.
After having contracted all vertices from $Q'$ into vertices from $Z'$, we obtain therefore a minor of $\inc(F[\tau])$.

Our objective will be to show that the treewidth of this minor is too large and arrive at a contradiction for $B$ being a \STBDS of $F$.

\begin{claim}\label{claim:kminor}
 $\inc(F[\tau])$ has a $K_{t+2}$-minor.
\end{claim}

\noindent
To prove the claim, we start with $H''$ and $Q''$ as copies of $H'$ and $Q'$, respectively.
 We use the invariant that every connected component of $H''[Z']$ is a minor of $\inc(F[\tau])$.

 For any part $A$ of the partition $P$, let $R_A$ denote the set of vertices $\{q\in Q'' : R(q)=A\}$.
 As long as $Q'' \neq \emptyset$, select a part $A$ of $P$ such that $|R_A|\ge 1$.
 Let $U:= \bigcup_{q\in R_A} N_{H''}(q)$.
 By the construction of $H'$ and $\cB(\cO_s)$ (Property (2)), we have that $|U|\ge \tnb+|R_A|-1$.

 If for every vertex $u\in U$, $|N_{H''}(u) \cap U| \ge \tnb$, then $H''[U]$ has at least $\tnb\cdot |U|/2 = \lceil 8 (t+2) \log (t+2) \rceil \cdot |U|$ edges.
 Then, by Theorem \ref{thm:mader},
 $H''[U]$ has a $K_{t+2}$-minor. By our invariant, $\inc(F[\tau])$ has a $K_{t+2}$-minor.

 Otherwise, there exist a vertex $z\in U$ such that $z$ has less than $\tnb$ neighbors in $U$.
 But then, merging $A$ into $z$ adds at least $|U|-\tnb+1 \ge |R_A|$ edges to $H''[Z']$.

 In the end, if no $K_{t+2}$-minor was found before $Q'' = \emptyset$, each merge of a part $A$ of $P$ into a vertex from $Z'$ added at least
 $|R_A|$ edges to $H''[Z']$. Therefore, the final graph $H''[Z']$ contains at least $|Q'|$ edges.
 By Claim \ref{claim:q}, $|Q'| \ge \frac{|Z| \cdot |\cO_s|}{3 \tnb t 2^{2k+1}}$ and $|\cO_s| = \ksame = 3 (\tnb)^2 t 2^{2k}$.
 Thus, $H''[Z']$ has at least $(8 (t+2) \log (t+2)) \cdot |Z'|$ edges.
 Consequently, $H''[Z']$ has a $K_{t+2}$-minor by Theorem \ref{thm:mader}, which is a minor of $\inc(F[\tau])$ by our invariant.
This proves Claim \ref{claim:kminor}.\hfill $\lrcorner$
\bigskip

\noindent
Claim \ref{claim:kminor} entails that $\inc(F[\tau])$ has treewidth at least $t+1$.
Since $\tau$ is a truth assignment to $B$, this is a contradiction to $B$ being a \STBDS of $F$.
This shows the correctness of Rule \ref{rule:NoMultiNb} and proves Lemma~\ref{lem:NoMultiNbSound}.
\end{proof}

\noindent
The number of choices the algorithm has in the nondeterministic steps is upper bounded by $\binom{\kobs}{k} \cdot \binom{2^k\cdot \ksame}{\ksame} \cdot k$,
and each series of guesses leads to a set $S$ of at most $6k \tnb$ variables.
Thus, the set $S^*$, the union of all such $S$, contains $2^{O(t^3\cdot k\cdot 4^k\cdot \mathsf{polylog}(t))}$ variables,
where $\mathsf{polylog}$ is a polylogarithmic function.
Concerning the running time, each \obstemp
is computed in time $O(n^2)$ by Lemma \ref{lem:obstemp} and their number is upper bounded by a constant.
The execution of Rule \ref{rule:MultiNb} and the construction of $\cB(\cO_s)$ need to compare the neighborhoods of a quadratic number of
vertices from $Q_m$. Since each vertex from $Q_m$ has a constant sized neighborhood, this can also be done in time $O(n^2)$.
Thus, the running time of the algorithm is quadratic.
This proves Lemma \ref{lem:wall}.

\subsection{The incidence graph has small treewidth}

\newcommand{\obstructionset}{\normalfont \textsf{obs}}

This subsection is devoted to the proof of Lemma \ref{lem:mso}.
 
We are going to use Arnborg \etal's
extension \cite{ArnborgLagergrenSeese91} of Courcelle's
Theorem~\cite{Courcelle90}. It gives, amongst others, a linear-time
algorithm that takes as input a graph $\AAA$ with labeled vertices and
edges, a tree decomposition of $\AAA$ of constant width, and a fixed
Monadic Second Order (MSO) sentence $\varphi(X)$, and computes a
minimum-sized set of vertices $X$ such that $\varphi(X)$ is true in
$\AAA$.

First, we define the labeled graph $\AAA_F$ for $F$.  The set of vertices
of $\AAA_F$ is $\lit(F) \cup \cla(F)$.  The vertices are labeled by
$\text{LIT}$ and $\text{CLA}$, respectively.  The vertices from
$\var(F)$ are additionally labeled by $\text{VAR}$.  The set of edges is
the union of the sets $\SB x \neg x \SM x\in \var(F) \SE$ and $\SB c \ell
\SM c\in \cla(F),$ $\ell \in \lit(c)\SE$, edges in
the first set are labeled $\text{NEG}$, and edges in the second set are
labeled $\text{IN}$.

Since a tree decomposition of $\AAA_F$ may be obtained from a tree
decomposition of $\inc(F)$ by replacing each variable by both its
literals, we have that $\tw(\AAA_F)\le 2 \cdot \tw(\inc(F))+1$ and we
obtain a constant-width tree decomposition of $\AAA_F$ in this way.

The goal is to find a minimum size subset $X$ of variables such that for each truth
assignment $\tau$ to $X$ the incidence graph of $F[\tau]$ belongs to
$\GGG_{\leq t}$, where $\GGG_{\leq t}$ denotes the class of all graphs of
treewidth at most~$t$. For testing membership in $\GGG_{\leq t}$ we use
a forbidden-minor characterization.  As proved in a series of papers by
Robertson and Seymour~\cite{RobertsonSeymour85}, every minor-closed
class $\GGG$ of graphs is characterized by a finite set
$\obstructionset(\GGG)$ of forbidden minors. That is,
$\obstructionset(\GGG)$ is a finite set of graphs such that a graph
$G$ belongs to $\GGG$ if and only if $G$ does not contain any graph
from $\obstructionset(G)$ as a minor. Clearly $\GGG_{\leq t}$ is minor-closed.
We denote its finite set of obstructions by
$\obstructionset(t)=\obstructionset(\GGG_{\leq t})$.
The set $\obstructionset(t)$ is explicitly given in~\cite{ArnborgProskurowskiCorneil90}
for $t\leq 3$ and it can be
computed in constant time
\cite{AdlerGK08,Lagergren98} for all other values of~$k$.

Next we are going to formulate an MSO sentence that checks whether for
each truth assignment $\tau$ to $X$, the incidence graph of $F[\tau]$
does not contain any of the graphs in $\obstructionset(t)$ as a minor.
We break up our MSO sentence into several simpler sentences and we use
the notation of \cite{FlumGrohe06}.

The following sentence checks whether $X$ is a set of variables.
\begin{quote}
$\text{var}(X) = \forall x (Xx \rightarrow \text{VAR} x)$  
\end{quote}
We associate a partial truth assignment to $X$ with a subset $Y$ of $\lit(F)$,
the literals set to $1$ by the partial truth assignment. This subset $Y$
contains no complementary literals, every literal in $Y$ is
either a variable from $X$ or its negation, and for every variable~$x$ from $X$, $x$ or $\neg x$ is in~$Y$.
The following sentence checks
whether $Y$ is an assignment to~$X$.
\begin{quote}
  $ \text{ass}(X,Y) = \forall y [Yy \rightarrow ((Xy \vee (\exists z (Xz
  \wedge \text{NEG} yz))) \wedge (\forall z (Yz \rightarrow \neg
  \text{NEG} yz)))]$
  
  \hfill $ \wedge \forall x [Xx \rightarrow (Yx \vee
  \exists y (Yy \wedge \text{NEG} xy))] $
\end{quote}
To test whether $\inc(F[\tau])$ has a graph $G$ with $V(G)=\{v_1,\dots,v_n\}$
as a minor, we will check whether it contains $n$ disjoint sets
$A_1,\dots,A_n$ of vertices, where each set $A_i$ corresponds to a vertex
$v_i$ of $G$, such that the following holds: each set $A_i$ induces a
connected subgraph in $\inc(F[\tau])$, and for every two vertices $v_i,v_j$ that are
adjacent in $G$, the corresponding sets $A_i,A_j$
are connected by an edge in $\inc(F[\tau])$. Deleting all vertices that are in none of the
$n$ sets, and contracting each of the sets into one vertex, one obtains
$G$ as a minor of $F[\tau]$. To test whether $\AAA_F$ has $G$ as a minor can be done
similarly, except that we need to ensure that each set $A_i$ is closed
under the complementation of literals (i.e., $x \in A_i$ iff $\neg x\in A_i$).

The following sentence checks whether $A$ is disjoint from~$B$.
\begin{quote}
 $\text{disjoint}(A,B) = \neg \exists x (Ax \wedge Bx)$  
\end{quote}
To check whether $A$ is connected,
we check that there is no set $B$ that is a proper nonempty
subset of $A$ such that $B$ is closed under taking neighbors in $A$.
\begin{quote}
$ \text{connected}(A) = \neg \exists B \; [\exists x (Ax \wedge \neg Bx)
\wedge \exists x (Bx) \wedge \forall x (Bx \rightarrow Ax)$ 

\hfill $\wedge \forall x,y ((Bx \wedge Ay \wedge (\text{IN}xy \vee \text{NEG}xy)) \rightarrow By)]$
\end{quote}
The next sentence checks whether $A$ is closed under complementation of
literals.
\begin{quote}
 $\text{closed}(A) = \forall x,y (\text{NEG}xy \rightarrow (Ax
 \leftrightarrow Ay))$  
\end{quote}
The following sentence checks whether some vertex from $A$ and some vertex from $B$ have a common edge labeled $\text{IN}$.
\begin{quote}
 $\text{edge}(A,B) = \exists x,y (Ax \wedge Bx \wedge \text{IN}xy)$  
\end{quote}
An assignment removes from the incidence graph all variables that are
assigned and all clauses that are assigned correctly.  Therefore, the
minors we seek must not contain any variable that is assigned nor any
clause that is assigned correctly.  The following sentence checks
whether all vertices from a set $A$ survive when assigning $Y$ to $X$.
\begin{quote}
  $\text{survives}(A,X,Y) = 
  \neg \exists x (Ax \wedge 
  (Xx 
  \vee
  \exists y (Xy \wedge \text{NEG}xy)
  \vee
  \exists y (Yy \wedge \text{IN}yx)))$
\end{quote}
We can now test whether a $G$-minor survives in the incidence graph as
follows:
\begin{quote}$ \text{$G$-minor}(X,Y) = \exists A_1, \dots, A_n [
  \bigwedge_{i=1}^n (\text{survives}(A_i) 
  \wedge   \text{connected}(A_i) 
  \wedge   \text{closed}(A_i)) $
 
  \hfill $ \wedge \bigwedge_{1\le i<j \le n} \text{disjoint}(A_i,A_j) \wedge
  \bigwedge_{1\le i< j \le n,\ v_iv_j\in E(G)} \text{edge}(A_i,A_j)] $
\end{quote}
Our final sentence checks whether $X$ is a strong $\TW_{\leq t}$\hy
backdoor set of~$F$.
\begin{quote}$
 \text{Strong}_t(X) = \text{var}(X) \wedge \forall Y [\text{ass}(X,Y) 
\rightarrow \forall_{G\in \obstructionset(t)} \neg \text{$G$-minor}(X,Y)))]$
\end{quote}
Recall that we assume $t$ to be a constant. Hence
$|\text{Strong}_t| = O(1)$. Moreover, the tree decomposition of $\AAA_F$ that
we described has width $O(1)$. We can
now use the result of Arnborg \etal~\cite{ArnborgLagergrenSeese91}
that provides a linear time algorithm for finding a smallest set $X$ of
vertices of $\AAA_F$ for which $\text{Strong}_t(X)$ holds.  This completes
the proof of Lemma~\ref{lem:mso}.

\section{Conclusion}

We have described a cubic-time algorithm solving SAT and \#SAT for
a large class of instances, namely those CNF formulas $F$
that have a strong backdoor set of size at most $k$ into
the class of formulas with incidence treewidth at most $t$,
where $k$ and $t$ are constants.
As illustrated in the introduction, this class of instances is larger than
the class of all formulas with bounded incidence treewidth.
We also designed an
approximation algorithm for finding an actual strong backdoor set.
Can our backdoor detection algorithm be improved to an exact algorithm? In other words,
is there an $O(n^c)$-time algorithm finding a $k$-sized \STBDS
of any formula $F$ with $\sb_t(F)\le k$
where $k,t$ are two constants and $c$ is an absolute constant
independent of $k$ and $t$? This question is even open for $t=1$.
An orthogonal question is how far one can generalize the class of
tractable (\#)SAT instances.

{\small
\bibliographystyle{plain}
\bibliography{literature}

\begin{thebibliography}{10}

\bibitem{AdlerGK08}
Isolde Adler, Martin Grohe, and Stephan Kreutzer.
\newblock Computing excluded minors.
\newblock In {\em Proceedings of the 19th Annual ACM-SIAM Symposium on Discrete
  Algorithms (SODA 2008)}, pages 641--650. SIAM, 2008.

\bibitem{AlekhnovichRazborov02}
Michael Alekhnovich and Alexander~A. Razborov.
\newblock Satisfiability, branch-width and {T}seitin tautologies.
\newblock In {\em Proceedings of the 43rd Annual IEEE Symposium on Foundations
  of Computer Science (FOCS 2002)}, pages 593--603, 2002.

\bibitem{ArnborgLagergrenSeese91}
Stefan Arnborg, Jens Lagergren, and Detlef Seese.
\newblock Easy problems for tree-decomposable graphs.
\newblock {\em J. Algorithms}, 12(2):308--340, 1991.

\bibitem{ArnborgProskurowskiCorneil90}
Stefan Arnborg, Andrzej Proskurowski, and Derek~G. Corneil.
\newblock Forbidden minors characterization of partial {$3$}-trees.
\newblock {\em Discrete Math.}, 80(1):1--19, 1990.

\bibitem{BacchusDalmaoPitassi03}
Fahiem Bacchus, Shannon Dalmao, and Toniann Pitassi.
\newblock Algorithms and complexity results for \#{SAT} and {B}ayesian
  inference.
\newblock In {\em Proceedings of the 44th Annual IEEE Symposium on Foundations
  of Computer Science (FOCS 2003)}, pages 340--351, 2003.

\bibitem{BidyukDechter07}
Bozhena Bidyuk and Rina Dechter.
\newblock Cutset sampling for {Bayesian} networks.
\newblock {\em J. Artif. Intell. Res.}, 28:1--48, 2007.

\bibitem{Bodlaender96}
Hans~L. Bodlaender.
\newblock A linear-time algorithm for finding tree-decompositions of small
  treewidth.
\newblock {\em SIAM J. Comput.}, 25(6):1305--1317, 1996.

\bibitem{Cook71}
Stephen~A. Cook.
\newblock The complexity of theorem-proving procedures.
\newblock In {\em Proceedings of the 3rd Annual ACM Symposium on Theory of
  Computing (STOC 1971)}, pages 151--158, 1971.

\bibitem{Courcelle90}
Bruno Courcelle.
\newblock Graph rewriting: an algebraic and logic approach.
\newblock In {\em Handbook of theoretical computer science, Vol.\ B}, pages
  193--242. Elsevier Science Publishers, North-Holland, Amsterdam, 1990.

\bibitem{CyganLPPS11}
Marek Cygan, Daniel Lokshtanov, Marcin Pilipczuk, Michal Pilipczuk, and Saket
  Saurabh.
\newblock On the hardness of losing width.
\newblock In D{\'a}niel Marx and Peter Rossmanith, editors, {\em Proceedings of
  the 6th International Symposium on Parameterized and Exact Computation (IPEC
  2011)}, volume 7112 of {\em Lecture Notes in Computer Science}, pages
  159--168. Springer, 2012.

\bibitem{Dechter03}
Rina Dechter.
\newblock {\em Constraint Processing}.
\newblock Morgan Kaufmann, 2003.

\bibitem{Diestel10}
Reinhard Diestel.
\newblock {\em Graph Theory}, volume 173 of {\em Graduate Texts in
  Mathematics}.
\newblock Springer Verlag, New York, 4th edition, 2010.

\bibitem{FischerMakowskyRavve06}
E.~Fischer, J.~A. Makowsky, and E.~R. Ravve.
\newblock Counting truth assignments of formulas of bounded tree-width or
  clique-width.
\newblock {\em Discr. Appl. Math.}, 156(4):511--529, 2008.

\bibitem{FlumGrohe06}
J\"{o}rg Flum and Martin Grohe.
\newblock {\em Parameterized Complexity Theory}, volume XIV of {\em Texts in
  Theoretical Computer Science. An EATCS Series}.
\newblock Springer Verlag, Berlin, 2006.

\bibitem{GaspersSzeider11b}
Serge Gaspers and Stefan Szeider.
\newblock Backdoors to satisfaction.
\newblock Technical Report 1110.6387, arXiv, 2011.

\bibitem{GaspersSzeider11a}
Serge Gaspers and Stefan Szeider.
\newblock Backdoors to acyclic {SAT}.
\newblock In {\em Proceedings of the 39th International Colloquium on Automata,
  Languages and Programming (ICALP 2012)}, Lecture Notes in Computer Science.
  Springer, 2012.
\newblock Available on arXiv 1110.6384.

\bibitem{GaspersSzeider12}
Serge Gaspers and Stefan Szeider.
\newblock Strong backdoors to nested satisfiabiliy.
\newblock In {\em Proceedings of the 15th International Conference on Theory
  and Applications of Satisfiability Testing (SAT 2012)}, Lecture Notes in
  Computer Science. Springer, 2012.
\newblock Available on arXiv 1202.4331.

\bibitem{GomesKautzSabharwalSelman08}
Carla~P. Gomes, Henry Kautz, Ashish Sabharwal, and Bart Selman.
\newblock Satisfiability solvers.
\newblock In {\em Handbook of Knowledge Representation}, volume~3 of {\em
  Foundations of Artificial Intelligence}, pages 89--134. Elsevier, 2008.

\bibitem{GroheKMW11}
Martin Grohe, Ken ichi Kawarabayashi, D{\'a}niel Marx, and Paul Wollan.
\newblock Finding topological subgraphs is fixed-parameter tractable.
\newblock In {\em Proceedings of the 43rd ACM Symposium on Theory of Computing,
  (STOC 2011)}, pages 479--488. ACM, 2011.

\bibitem{Knuth90}
Donald~E. Knuth.
\newblock Nested satisfiability.
\newblock {\em Acta Informatica}, 28(1):1--6, 1990.

\bibitem{Lagergren98}
Jens Lagergren.
\newblock Upperbounds on the size of obstructions and intertwines.
\newblock {\em Journal of Combinatorial Theory, Series B}, 73(1):7--40, 1998.

\bibitem{Levin73}
Leonid Levin.
\newblock Universal sequential search problems.
\newblock {\em Problems of Information Transmission}, 9(3):265--266, 1973.

\bibitem{LiB11}
Zijie Li and Peter van Beek.
\newblock Finding small backdoors in {SAT} instances.
\newblock In {\em Proceedings of the 24th Canadian Conference on Artificial
  Intelligence (Canadian AI 2011)}, volume 6657 of {\em Lecture Notes in
  Computer Science}, pages 269--280. Springer, 2011.

\bibitem{Mader68}
W.~Mader.
\newblock Homomorphies{\"a}tze f{\"u}r {G}raphen.
\newblock {\em Mathematische Annalen}, 178:154--168, 1968.

\bibitem{Marx08b}
D{\'a}niel Marx.
\newblock Parameterized complexity and approximation algorithms.
\newblock {\em The Computer Journal}, 51(1):60--78, 2008.

\bibitem{NishimuraRagdeSzeider04-informal}
Naomi Nishimura, Prabhakar Ragde, and Stefan Szeider.
\newblock Detecting backdoor sets with respect to {Horn} and binary clauses.
\newblock In {\em Proceedings of the 7th International Conference on Theory and
  Applications of Satisfiability Testing (SAT 2004)}, pages 96--103, 2004.

\bibitem{NishimuraRagdeSzeider07}
Naomi Nishimura, Prabhakar Ragde, and Stefan Szeider.
\newblock Solving \#{S}{A}{T} using vertex covers.
\newblock {\em Acta Informatica}, 44(7-8):509--523, 2007.

\bibitem{RazgonOSullivan09}
Igor Razgon and Barry O'Sullivan.
\newblock Almost 2-{SAT} is fixed parameter tractable.
\newblock {\em Journal of Computer and System Sciences}, 75(8):435--450, 2009.

\bibitem{RobertsonSeymour85}
Neil Robertson and P.~D. Seymour.
\newblock Disjoint paths---a survey.
\newblock {\em SIAM J. Algebraic Discrete Methods}, 6(2):300--305, 1985.

\bibitem{RobertsonSeymour86}
Neil Robertson and P.~D. Seymour.
\newblock Graph minors. {II}. {A}lgorithmic aspects of tree-width.
\newblock {\em J. Algorithms}, 7(3):309--322, 1986.

\bibitem{RobertsonSeymour86b}
Neil Robertson and P.~D. Seymour.
\newblock Graph minors. {V}. {E}xcluding a planar graph.
\newblock {\em J. Combin. Theory Ser. B}, 41(1):92--114, 1986.

\bibitem{RobertsonSeymour91}
Neil Robertson and P.~D. Seymour.
\newblock Graph minors~{X}. {O}bstructions to tree-decomposition.
\newblock {\em J. Combin. Theory Ser. B}, 52(2):153--190, 1991.

\bibitem{RobertsonSeymourThomas94}
Neil Robertson, Paul Seymour, and Robin Thomas.
\newblock Quickly excluding a planar graph.
\newblock {\em J. Combin. Theory Ser. B}, 62(2):323--348, 1994.

\bibitem{SamerSzeider10}
Marko Samer and Stefan Szeider.
\newblock Algorithms for propositional model counting.
\newblock {\em J. Discrete Algorithms}, 8(1):50--64, 2010.

\bibitem{Schaefer78}
Thomas~J. Schaefer.
\newblock The complexity of satisfiability problems.
\newblock In {\em Proceedings of the 10th Annual ACM Symposium on Theory of
  Computing (STOC 1978)}, pages 216--226. ACM, 1978.

\bibitem{Spinrad03}
Jeremy~P. Spinrad.
\newblock {\em Efficient Graph Representations}.
\newblock Fields Institute Monographs. AMS, 2003.

\bibitem{Thomason01}
Andrew Thomason.
\newblock The extremal function for complete minors.
\newblock {\em Journal of Combinatorial Theory, Series B}, 81(2):318--338,
  2001.

\bibitem{Valiant79b}
L.~G. Valiant.
\newblock The complexity of computing the permanent.
\newblock {\em Theoretical Computer Science}, 8(2):189--201, 1979.

\bibitem{WilliamsGomesSelman03}
Ryan Williams, Carla Gomes, and Bart Selman.
\newblock Backdoors to typical case complexity.
\newblock In {\em Proceedings of the 18th International Joint Conference on
  Artificial Intelligence (IJCAI 2003)}, pages 1173--1178, 2003.

\end{thebibliography}
}

\end{document}